\documentclass[conference]{IEEEtran}
\usepackage{graphicx}
\usepackage{dcolumn}
\usepackage{bm}
\usepackage{hyperref}
\usepackage{physics}
\usepackage{booktabs}
\usepackage{amsthm}
\usepackage{array} % Required for fixed-width columns and better table alignment
\usepackage{amsmath}  % Provides better math environments
\usepackage{amsfonts} % For math symbols
\usepackage{comment}
\usepackage{subcaption}
\usepackage{algorithm}
\usepackage{algpseudocode}
\usepackage{tabularx}
\usepackage{booktabs}
\usepackage{amsmath}

\usepackage{graphicx}
\usepackage{tikz}
\usetikzlibrary{quantikz2,shapes.geometric,shapes.arrows, arrows.meta, positioning, calc, backgrounds} % Works with standard quantikz as well

\theoremstyle{remark} 

\newtheorem{remark}{Remark} 

\newcommand{\Algo}{AQ-Stacker}
\newcommand{\AOD}{Algebraic Operator Decomposition}
\newcommand{\decot}{Hardware-dependent Coherence Threshold}
\newtheorem{corollary}{Corollary}

\begin{document}

\title{Algebraic Operator Decomposition: A Partitioned Architecture for Noise-Resilient Quantum Computing}
%\author{Wladimir Silva}
%\affiliation{Department of Electrical and Computer Engineering, North Carolina State University}
\author{
    \IEEEauthorblockN{Wladimir Silva}
    \IEEEauthorblockA{
        \textit{Department of Electrical and Computer Engineering} \\
        \textit{North Carolina State University}\\
        Raleigh, USA \\
        wsilva@ncsu.edu
    }
}

\date{\today}
\maketitle
\begin{abstract}
%We present a potentially useful operator-decomposition architecture that converts circuit-depth-induced physical error into a controllable combination of local quantum errors plus classical sampling overhead using {\AOD}
%with bound $|\langle T\rangle_\Lambda-\langle T\rangle_{\rm ideal}| \leq \sum_k |c_k|\,\text{local\_error}_k$
%for a global linear operator $T \in \mathcal{L}(V, W)$ decomposed into $K$ independent sub-channels 
%where $\mathcal{L}(\mathcal{V}, \mathcal{W})$ is the space of all linear operators mapping $\mathcal{V} \to \mathcal{W}$.
We present an operator-decomposition architecture that mathematically maps a global operator into independently executable local operators, reducing the maximum quantum circuit depth at the cost of classical reconstruction and sampling overhead.
By framing complex Quantum Circuits around an operator in a vector space that can be algebraically pre-decomposed,
AOD complements quantum error correction and error-mitigation approaches by performing algebraic decomposition before quantum execution.
%we introduce a third path in quantum error engineering: unlike Quantum Error Correction (QEC) we are not fixing errors physically,
%nor statistically post-execution (like traditional error mitigation), but algebraically trying to avoid them pre-execution.

Our approach leans in the computer science definition of a Monoid: a design pattern and mathematical concept consisting of a data type, a combining function that is associative, and a safe identity (neutral) element that does not change other values when combined.
Simulation wise we define a MapReduce programming model where
the \textit{addition ($+$)} is the reducer, thus leveraging a naturally stable commutative monoid which carries zero "negative-probability tax" or phase conflicts.

Furthermore, we define a Vector Space of Linear Operators over Additive Abelian Groups that benefit from this paradigm, including:
Inner Products, Series expansions, Traces and Convolutions.
Finally, we present the mathematical foundations and simulation results for this paradigm.
\end{abstract}

\maketitle

\section{Introduction}
Algebraic Operator Decomposition (AOD) is a technique used to break a complex linear operator (or matrix) into a sum, product, or combination of simpler, algebraically structured components. 
Its goal is to simplify solving systems of equations, finding operator functions, or understanding geometric transformations.
In the context of Computer Science, specifically distributed computing, AOD can be thought as a form of MapReduce, a programming model designed to process massive datasets in parallel across clusters \cite{dean2004mapreduce}.

While AOD-MapReduce is designed to speed up complex processes by divide-and-conquer,
in NISQ Quantum Computing, it can be used as a proactive method to mitigate noise and manage decoherence,
which limit circuit depth, reduce calculation accuracy, and cause quantum states to decay \cite{Preskill2018nisq}.
%While "MapReduce" is a classical distributed-computing optimization method for data scaling \cite{dean2004mapreduce},
%our framework acts as a spatial error-mitigation strategy to bypass the physical limitations of the NISQ era \cite{Preskill2018nisq}.
Our work illustrates how AOD-MapReduce can be used not to correct noise (QEC) or manage decoherence (via Circuit Cutting-Knitting), 
but to avoid both at the same time, thus acting as a spatial error-mitigation strategy to bypass the limitations of the NISQ era.
The main contributions of this manuscript include:
\begin{itemize}
    \item \textbf{Algebraic Operator Decomposition (AOD) Paradigm:} Proposes a novel proactive pre-execution error-avoidance framework that acts as a spatial error-mitigation strategy by algebraically fracturing complex quantum circuits prior to execution.
    \item \textbf{Quantum MapReduce Programming Model:} Formulates a hybrid distributed computing scheme structured around a naturally stable commutative monoid under classical addition, preventing phase conflicts and negative-probability overhead.
    %\item \textbf{Cross-Domain Applicability:} Demonstrates robust functional extension of the AOD pipeline across distinct non-trivial computational domains: multi-qubit inner products (Hadamard Tests), discrete spatial differential operators ($\nabla^2$), and non-linear activation loops ($\tanh$, GELU).
	\item \textbf{Cross-Domain Applicability:} We demonstrate the AOD pipeline across distinct computational domains: multi-qubit inner products (Hadamard Test), discrete spatial differential operators ($\nabla^2$), and non-linear activation loops ($\tanh$, GELU).
    %\item \textbf{Rigorous Hardware-Noise Tracking:} Formalizes the theoretical limit of the physical ``Decoherence Threshold'' ($D \approx 500$) and mathematically proves both Spatial Error Isolation (zero downstream gate error propagation) and the Critical Efficiency Threshold over noisy monolithic registers.
	\item \textbf{Hardware-Dependent Noise and Partition Analysis:} We formulate a depth-dependent depolarizing noise model and analyze the tradeoff between circuit depth, local noise, partition size, and sampling overhead.
\end{itemize}

We begin by quantifying a custom \textit{Depolarizing Noise Model} based on circuit depth and its correlated \textit{{\decot}}.

\subsection{Cumulative Gate Noise and the {\decot}}
We first evaluate AOD under a simplified depth-dependent depolarizing model. This model is intended to isolate the relationship between circuit depth $(D)$ and signal attenuation using cumulative gate noise $(\lambda)$ and the average readout error $(\epsilon_r)$; it does not model coherent errors, leakage, crosstalk, non-Markovian noise, or gate-dependent correlated errors. 
The analysis therefore establishes behavior under the specified independent depolarizing model rather than a general physical noise model.
Let $\epsilon_{1Q}$, $\epsilon_{2Q}$ be the average 1-qubit, 2-qubit error rates of an arbitrary QPU, 
then the cumulative gate noise ($\lambda$) is defined as follows:

$$\lambda = 1 - (1 - \epsilon_{1Q})^{D} \times (1 - \epsilon_{2Q})^{D}$$ 

Where $(1 - \epsilon_{1Q})^{D} \times (1 - \epsilon_{2Q})^{D}$ is the probability of a state surviving all gates intact. 
For a complex Quantum Circuit with large depth and accounting for a hardware-dependent depth/noise threshold $\approx 500$ under the selected parameters
The process works by:
\begin{itemize}
	\item Splitting the problem into $K$ smaller partial sub-vectors.
	\item Executing the shallow circuits in a Multi-QPU setup.
	\item Classically reducing the outcomes to reconstruct the original result.
\end{itemize}

\begin{remark}[]
The noise parameters $\epsilon_{1Q}$, $\epsilon_{2Q}$, and $\epsilon_r$ are fetched from the \textit{ibm\_kingston} QPU via REST \cite{ibm_quantum_rest_api}.
The {\decot} is a \textbf{modeling estimate} at which the information dissipates and introduced in our prior work on the {\Algo} algorithm \cite{silva2024aqstacker} as: $D_{max} \approx T_2/T_{gate}$ where $T_2 \le 2T_1$ with $T_2 \approx 250Kns$ (Coherence Time), and $T_{gate} \approx 500ns$ (2-Qubit Gate Time) for the average performance metrics of the Latest IBM Heron r3 family of QPUs \cite{ibmquantum2026}.
\end{remark}

\subsection{The Architectural Lineage: From Operators to Vector Spaces}

We consider all processes that belong to a Vector Space (or in functional analysis, a Banach Space or Hilbert Space \cite{kreyszig1978introductory}), where the operations themselves act as Linear Operators (or Linear Functionals)
belonging to an Additive Abelian Group under addition (See Table \ref{tab:processes_ paradigm}).
Because our method scales the components by classical coefficients, we can put them under a Vector Space of Functions/Operators which satisfies two fundamental linear properties:
%\subsubsection{Vector Space of Linear Operators}
%This algebraic space satisfies two fundamental linear properties:
\begin{itemize}
	\item Cauchy's Functional Equation~\cite{Cauchy1821}: For any linear operator $T$, the rule is strictly: $T(x + y) = T(x) + T(y)$
		\begin{enumerate}
			\item \textbf{Map:} Split the global input $X$ into fragments ($x_1 + x_2 + \dots$), and evaluate $T(x_i)$ independently.
			\item \textbf{Reduce:} Sum the results classically ($\sum T(x_i)$).
		\end{enumerate}
	\item The Riesz Representation Theorem: In a Hilbert space, every continuous linear functional $T(x)$ can be written as an inner product with a specific vector $y$: $T(x) = \langle y, x \rangle$~\cite{delrio2017directprooffriesz}.
	%Because all operations are framed as linear functionals, they can all be mapped back to an inner product.
	Many of the linear observables considered in this work admit inner-product representations, which makes Hadamard-test-based primitives a convenient execution mechanism.
\end{itemize}

\begin{table}[htbp]
\centering
\caption{Process Vector Space under Abelian Addition. It describes algebraic decomposability; it does not by itself guarantee an efficient quantum implementation or quantum advantage.}
\label{tab:processes_ paradigm}
\begin{tabularx}{\linewidth}{>{\raggedright\arraybackslash}p{2.3cm}X X}
\toprule
\textbf{Process} & \textbf{How it splits Linearly (The ``Map'')} & \textbf{How it recombines (The ``Reduce'')} \\
\midrule
Inner Products & $(\phi, x_1 + x_2)$ & $\langle \phi, x_1 \rangle + \langle \phi, x_2 \rangle$ \\
\addlinespace
Series Expansions & Polynomial chunks: $\sum c_k x^k$ & Summing individual scalar terms \\
\addlinespace
Traces ($\text{Tr}(A + B)$) & Splitting matrices into sub-blocks & $\text{Tr}(A) + \text{Tr}(B)$ \\
\addlinespace
Convolutions & Splitting an image into sliding patches & Summing localized patch overlaps \\
\bottomrule
\end{tabularx}
\end{table}

\begin{table*}[htbp]
\centering
%\caption{HW Depth (plus CNOT count) and partitioning trade-offs for $N=256$ features of the Inner Product (using the Hadamard Test) from ibm\_kingston with error rates: $\epsilon_{1Q}=0.0002$, $\epsilon_{2Q}=0.0019$, $\epsilon_{r}=0.0076$}
\caption{\textbf{Partition-depth-noise tradeoff under the selected hardware model} for $N=256$ features of the Inner Product (using the Hadamard Test) for ibm\_kingston with error rates: $\epsilon_{1Q}=0.0002$, $\epsilon_{2Q}=0.0019$, $\epsilon_{r}=0.0076$.}
\label{tab:circuit_complexity_mapreduce}
\begin{tabular}{ccccccc}
\hline
\textbf{Qubits ($n$)} & \textbf{Partitions ($k$)} & \textbf{Qu. / Part. ($n+1$)} & \textbf{Part. Qu.} & \textbf{HW ($D$)} & \textbf{CNOT-c} & \textbf{Noise ($\lambda$)} \\ \hline
1                     & 128                       & 2                            & 256                 & 10                      & 2                   & 0.0208 \\
2                     & 64                        & 3                            & 192                 & 61                      & 20                  & 0.1203 \\
3                     & 32                        & 4                            & 128                 & 180                     & 70                  & 0.3150 \\
4                     & 16                        & 5                            & 80                  & 460                     & 185                 & 0.6197 \\
5                     & 8                         & 6                            & 48                  & 1029                    & 422                 & 0.8850 \\
6                     & 4                         & 7                            & 28                  & 2227                    & 908                 & 0.9907 \\
7                     & 2                         & 8                            & 16                  & 4677                    & 1913                & 0.9999 \\
8                     & 1                         & 9                            & 9                   & 9472                    & 3893                & 1.0000 \\ \hline
\end{tabular}
\end{table*}

%---- comment
\begin{comment}
\begin{table*}[htbp]
\centering
\caption{Circuit Depth (plus CNOT count) and Partitioning Trade-offs for a $N=256$ Feature Space.}
\label{tab:circuit_complexity_mapreduce}
\begin{tabular}{cccccc}
\hline
\textbf{Data Qubits ($n$)} & \textbf{Partitions ($k$)} & \textbf{Asymptotic} & \textbf{HW} & \textbf{CNOT-c} \\ \hline
1                     & 128                       & 5                         & 10                      & 2                   \\
2                     & 64                        & 12                        & 61                      & 20                  \\
3                     & 32                        & 26                        & 180                     & 70                  \\
4                     & 16                        & 56                        & 460                     & 185                 \\
5                     & 8                         & 118                       & 1029                    & 422                 \\
6                     & 4                         & 244                       & 2227                    & 908                 \\
7                     & 2                         & 498                       & 4677                    & 1913                \\
8                     & 1                         & 1008                      & 9472                    & 3893                \\ \hline
\end{tabular}
\end{table*}
\end{comment}
%---- end comment

\section{Existing Methods}
In the rapidly evolving field of quantum information science, new paradigms include Quantum Circuit Knitting \cite{Piveteau_2024}, Quasiprobability Decomposition (QPD) \cite{Piveteau_2022}, and Randomized Quantum Linear Algebra \cite{Wang_2024} among others.
Within these paradigms, some of its processes include:
\begin{itemize}
	\item \textbf{High-Dimensional Tensor Contractions (Wire \& Gate Cutting):} When moving from flat vectors to multidimensional tensors, full quantum state tracking triggers an exponential depth explosion \cite{Matthews_2018}.
	\item \textbf{Linear Combinations of Unitaries (LCUs) \& Block Encodings:} block-encoding a matrix $A$ inside a larger unitary usually requires massive ancilla arrays and deep control logic \cite{Childs2012Hamiltonian}.
	\item \textbf{Matrix Exponentiation \& Spectral Properties (Pauli Decompositions)}: Common in quantum chemistry and machine learning, to evaluate functions of arbitrary dense matrices.
	They generally require exponential time because the number of Pauli strings scales as $4^n$ for $n$ qubits, depending on whether the target operator or matrix is dense, sparse, or has a specific structural property \cite{Georges_2025}.
\end{itemize}
These processes tend to be reactive and monolithic in nature and can run in exponential times or require complex processes like state tomography \cite{Gily_n_2019, Georges_2025}.
\textit{Our paradigm, on the other hand, aims to be a proactive algebraic divide and conquer strategy against both: noise and decoherence.}

\section{Mathematical Background}
Let $\mathcal{V}$ and $\mathcal{W}$ be vector spaces over a field $\mathbb{F}$ (such as $\mathbb{R}$ or $\mathbb{C}$), and let $\mathcal{L}(\mathcal{V}, \mathcal{W})$ be the space of all linear operators mapping $\mathcal{V} \to \mathcal{W}$.

\subsection{{\AOD} Framework}
Any global operator $T \in \mathcal{L}(\mathcal{V}, \mathcal{W})$ that can be classically factored into a linear combination of localized, bounded basis operations $\tau_k$ can be evaluated via a distributed sub-coherence quantum channel architecture:
$$T = \sum_{k=1}^K c_k \tau_k \quad \text{where } c_k \in \mathbb{F}$$ 

\begin{enumerate}
	\item \textbf{Map:} The host computer decomposes the global mathematical expression into $K$ independent fragments and extracts the structural coefficients $c_k$.
	$$T\rightarrow\{c_k,\tau_k\}_{k=1}^K$$
	\item \textbf{Compute:} Each isolated basis operation $\tau_k$ is compiled into a dedicated, shallow quantum primitive whose circuit depth $D(\tau_k)$ strictly respects the physical hardware coherence threshold: $\max_{k} D(\tau_k) \ll D_{\max}$.
	$$\tau_k\rightarrow\widehat{\langle\tau_k\rangle}$$
	\item \textbf{Reduce:} The scalar measurement outcomes $\langle \tau_k \rangle$ are collected and synthesized via a commutative monoid (classical addition) to reconstruct the global operator property:  
	%$\langle T \rangle = \sum_{k=1}^K c_k \langle \tau_k \rangle$
	$$\widehat{\langle T\rangle} = \sum_kc_k\widehat{\langle\tau_k\rangle}$$
\end{enumerate}

\begin{remark}[]
By shifting to general vector space operators, we formalize two explicit structural hardware protections:
	
	\textbf{1) Depth Reduction Under Partitioning:} In a global monolithic circuit $T$, the global depth scales with system size~\cite{Ithier_2005}, $D_{\text{mono}} = f(N)$, inevitably breaching the {\decot} ($D_{\text{mono}} > D_{\max}$).
	In our framework, because the algebra is fractured prior to runtime, the quantum execution time is bounded entirely by the deepest sub-component:
$D_{\text{actual}} = \max_{k} D(\tau_k) < D_{\max}$.
	The circuit completes execution and reads out its state before physical phase-decoherence or $T_2$ relaxation can corrupt the system.
	
	\textbf{2) Spatial Error Containment:}  In a monolithic quantum circuit, a single gate error ($\epsilon_g$) at step $t$ entangles with the rest of the register, causing exponential fidelity decay across the entire downstream state vector: $F \approx (1-\epsilon_g)^{n \cdot d}$. \cite{chiesa2018spatialisolationimplieszero,wang2026scalablequantumerrormitigation}
	%Under Algebraic Pre-Decomposition, the tensor products are broken into smaller units: An error occurring during the execution of sub-channel $\tau_i$ is confined to that specific slice; it possesses no mathematical pathway to propagate to or corrupt the data of sub-channel $\tau_j$.
	Under Algebraic Pre-Decomposition, the tensor products are broken into smaller units: An error occurring during the execution of sub-channel $\tau_i$ is confined to that specific slice.
	Because the subchannels are executed independently and no quantum information is exchanged between them, an error occurring during execution of $\tau_i$ does not directly alter the quantum state or measurement process of $\tau_j$.
	The error can nevertheless affect the final reconstructed observable through the coefficient $c_i$.
\end{remark}

\subsection{Concrete Extensions Beyond the Inner Product}
We consider four scenarios on functional decomposition traditionally difficult for NISQ hardware:

\begin{enumerate}
	\item \textbf{Two-State Hadamard Test:} For finding the Inner Product $\langle\phi|\psi\rangle$ of two real valued vectors $\vec{x}$, $\vec{w}$ in the context of linear transformations for Machine Learning \cite{imamura2026parallelhadamardtest}.
	\item \textbf{Differential Matrix Operators ($\nabla^2, \partial_t$):} For solving partial differential equations, global finite-difference stencils can be fractured into localized coordinate-difference primitives, mapped to ultra-shallow 2-qubit phase channels, and classically accumulated to yield global gradients~\cite{parr2018matrixcalculusneeddeep}.
	\item \textbf{Non-Linear Functional Approximations:} Functions like $f(x) = \tanh(x)$ or $\text{GELU}(x)$ cannot be natively run as unitaries. Classically mapping them to a Chebyshev polynomial series ($\sum c_k T_k(x)$) \cite{Rao_2023} allows each polynomial degree $T_k$ to be handled by independent, ultra-shallow single-qubit rotation channels, combining them cleanly at the reduction stage.
	\item \textbf{Convolutions:} A convolution can be expressed as a series of inner products (or dot products) between a shifted, mirrored version of a filter and an input signal. In deep learning and signal processing, this operation measures the similarity between the filter and different local regions of the input \cite{NIPS2012_c399862d, neyshabur2020learningconvolutionsscratch}.
\end{enumerate}
\section{Method Description}

\subsection {Scenario 1: Linear Transformations via Inner Product}
For two real-valued vectors $\vec{x}$ and $\vec{w}$ and their corresponding L2 normalizations
$|\psi\rangle = \frac{\vec{x}}{\|\vec{x}\|}, \quad |\phi\rangle = \frac{\vec{w}}{\|\vec{w}\|}$
The classical dot product $\vec{x} \cdot \vec{w} = \sum_{i=1}^n x_i w_i$ is defined as 
$\vec{x} \cdot \vec{w} = \langle \psi | \phi \rangle \cdot \|\vec{x}\| \cdot \|\vec{w}\|, \vec{x}, \vec{w} \in \mathcal{R}$.
The matrix product part of a linear transformation $Z = XW + B$ is the structural collection of row-column dot products: $C_{ij} = (\text{Row}_i A) \cdot (\text{Column}_j B)$.
\begin{itemize}
	\item \textbf{Map:} For n-qubits, split the N-feature vectors into $K=N/2^n$ partial sub-vectors (e.g., $K=256/2=128$ for MNIST-16 with 1-qubit partial Hadamard Tests).
	For each chunk $k \in [1,K]$, let two partial vectors be: $\mathbf{a}^{(k)}$ and $\mathbf{b}^{(k)}$.
	Calculate the partial norms: $N_a^{(k)} = \Vert{}\mathbf{a}^{(k)}\Vert{}_2$ and $N_b^{(k)} = \Vert{}\mathbf{b}^{(k)}\Vert{}_2$.
	\item \textbf{Compute:} Execute the shallow quantum circuits (Hadamard Tests) in a Multi-QPU environment.
	The normalized vectors for each QuantumCircuit will be: $\tilde{\mathbf{a}}^{(k)} = \mathbf{a}^{(k)} / N_a^{(k)}$ and $\tilde{\mathbf{b}}^{(k)} = \mathbf{b}^{(k)} / N_b^{(k)}$.
	The returned partial inner products will be: $d_k = Re\langle \tilde{\mathbf{b}}^{(k)} \vert{} \tilde{\mathbf{a}}^{(k)} \rangle$.
	\item \textbf{Reduce:} Classically sum the partial results.
	Scale the quantum result by the norms: $D_k = d_k \cdot (N_a^{(k)} \cdot N_b^{(k)})$.
	Sum the results: $D_{total} = \sum_{k=1}^{K} D_k$.
\end{itemize}

We use the quantum circuit primitive (Fig. \ref{fig:hadamard_test}) from the AQ-Stacker algorithm described in \cite{silva2024aqstacker} where $Re\langle\phi|\psi\rangle = P(0) - P(1)$.
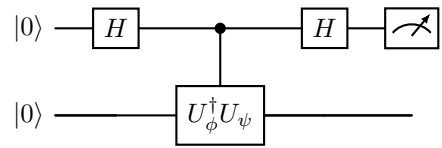
\begin{figure}[htbp]
	\centering
	\begin{quantikz}
		\lstick{\text{\ensuremath{\ket{0}}}} & \gate{H} & \ctrl{1} & \gate{H} & \meter{} \\
		\lstick{\text{\ensuremath{\ket{0}}}} & \qw      & \gate{\text{\ensuremath{U_\phi^\dagger U_\psi}}} & \qw      & \qw
	\end{quantikz}
	\caption{Quantum circuit for the Hadamard Test of two states $\psi, \phi$ in unitary form.} 
	\label{fig:hadamard_test} 
\end{figure}
%----
\begin{comment}
\begin{figure}[htbp]
	\centering
	\begin{quantikz}
		\lstick{\text{\ensuremath{\ket{0}}}} & \gate{H} \slice[style=dashed]{\text{\footnotesize\ensuremath{\ket{\psi_1}}}} & \ctrl{1} \slice[style=dashed]{\text{\footnotesize\ensuremath{\ket{\psi_2}}}} & \gate{H} \slice[style=dashed]{\text{\footnotesize\ensuremath{\ket{\psi_{\text{final}}}}}} & \meter{} \\
		\lstick{\text{\ensuremath{\ket{0}}}} & \qw                                                                          & \gate{\text{\ensuremath{U_\phi^\dagger U_\psi}}}                             & \qw                                                                                      & \qw
	\end{quantikz}
	\caption{Quantum circuit for the Hadamard Test of two states $\psi, \phi$ in unitary form.} 
	\label{fig:hadamard_test} 
\end{figure}
\end{comment}
%----

\subsection{Scenario 2: Differential Matrix Operators}
For a discrete function or physical field $U$ evaluated over an $N$-element spatial grid, global differential operators such as the Laplacian ($\nabla^2$) or spatial gradients ($\partial_x$) are traditionally represented as large, sparse stencil matrices. 
%Compiling these as monolithic quantum operators yields deep circuits vulnerable to exponential gate error propagation. 
Compiling these as monolithic quantum operators yields deep circuits where accumulated error increases rapidly with depth and system size under the independent-gate noise model adopted here. 
AOD fractures the global stencil into localized, independent coordinate-difference primitives.

\begin{itemize}
    \item \textbf{Map:} Let the global differential operator be denoted as $\mathcal{D} \in \mathcal{L}(V, W)$. We decompose $\mathcal{D}$ into a linear combination of $K$ localized coordinate-difference primitives $\tau_k$, such that:
    \begin{equation}
        \mathcal{D} = \sum_{k=1}^{K} c_k \tau_k
    \end{equation}
    where $c_k \in \mathbb{R}$ are the physical finite-difference scale coefficients (e.g., involving grid spacing factors $1/h$ or $1/2h$). For a standard central-difference gradient, each $\tau_k$ isolates a compact, local neighborhood consisting of adjacent coordinate nodes $u_i$ and $u_{i+1}$.
    
    \item \textbf{Compute:} Instead of encoding the entire field state vector globally, each localized coordinate patch is evaluated independently. The neighboring node amplitudes are mapped to an ultra-shallow 2-qubit phase channel. Because each sub-circuit isolates a small target register, the localized differential step:
    \begin{equation}
        g_k = \text{Re}\langle \psi^{(k)} | \tau_k | \psi^{(k)} \rangle
    \end{equation}
    is executed using a shallow interferometric primitive. The circuit depth is strictly bounded by $\max_k D(\tau_k) \ll D_{\text{max}}$, preventing any spatial gate error from propagating outside the local 2-qubit coordinate boundary.
    
    \item \textbf{Reduce:} Collect the scalar measurement outcomes $g_k$ representing localized gradients. The global differential field or gradient expectation value $\langle \mathcal{D} \rangle$ is reconstructed by classically executing the linear summation over the commutative monoid:
    \begin{equation}
        \langle \mathcal{D} \rangle = \sum_{k=1}^{K} c_k g_k
    \end{equation}
    This shifts the entire structural burden of the spatial stencil assembly from deep quantum coherence to a classical addition reducer.
\end{itemize}

% circuit
\begin{figure}[htbp]
	\centering
	\resizebox{\columnwidth}{!}{%
		\begin{quantikz}
			\lstick{ancilla:} & \gate{H} \slice[style=dashed]{\text{\footnotesize\ensuremath{\lvert\Psi_1\rangle}}} & \ctrl{1} \slice[style=dashed]{\text{\footnotesize\ensuremath{\lvert\Psi_2\rangle}}} & \gate{H} \slice[style=dashed]{\text{\footnotesize\ensuremath{\lvert\Psi_{\text{final}}\rangle}}} & \meter{} \\
			\lstick{spatial\_node:} & \gate{\text{\ensuremath{U(\theta,0,0)}}} & \gate[style={fill=black, minimum width=3mm, minimum height=3mm}]{} & \qw & \qw
		\end{quantikz}%
	}
	\caption{Coordinate-difference circuit primitive detailing spatial state vector validation boundaries.}
	\label{fig:diff_matrix}
\end{figure}
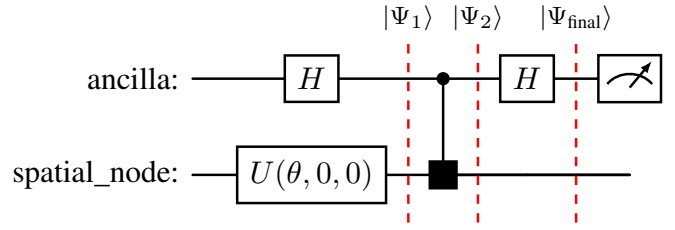

In Fig. \ref{fig:diff_matrix} we track the composite state evolution of the system register $\lvert \Psi \rangle$ across each discrete operational slice:
\begin{align}
    \lvert \Psi_0 \rangle &= \lvert 0 \rangle \otimes \lvert \psi^{(k)} \rangle \\
    \lvert \Psi_1 \rangle &= \frac{1}{\sqrt{2}}\bigl(\lvert 0 \rangle + \lvert 1 \rangle\bigr) \otimes \lvert \psi^{(k)} \rangle \\
    \lvert \Psi_2 \rangle &= \frac{1}{\sqrt{2}}\Bigl( \lvert 0 \rangle \otimes \lvert \psi^{(k)} \rangle + \lvert 1 \rangle \otimes \tau_k \lvert \psi^{(k)} \rangle \Bigr) \\
    \lvert \Psi_{\text{final}} \rangle &= \frac{1}{2} \lvert 0 \rangle \otimes \bigl(I + \tau_k\bigr)\lvert \psi^{(k)} \rangle + \frac{1}{2} \lvert 1 \rangle \otimes \bigl(I - \tau_k\bigr)\lvert \psi^{(k)} \rangle
\end{align}
Measuring the auxiliary ancilla register in the computational basis yields the targeted localized gradient primitive metric:
\begin{equation}
    g_k = P(0) - P(1) = \text{Re}\langle \psi^{(k)} \vert \tau_k \vert \psi^{(k)} \rangle
\end{equation}

\subsection{Scenario 3: Non-Linear Functional Approximations}
Non-linear activation functions and physical potentials—such as $f(x) = \tanh(x)$ or $\text{GELU}(x)$—cannot be directly implemented as native quantum unitaries. Traditional approaches require deep ancilla-driven block encodings or quantum state tomography, which rapidly hit the decoherence threshold. Under the AOD framework, we approximate the non-linear functional mapping as a bounded linear combination of orthogonal polynomials, processing each degree independently.

\begin{itemize}
    \item \textbf{Map:} Let the targeted non-linear functional approximation $f(x)$ be defined over a compact interval $[-1, 1]$. We project the non-linearity onto a Chebyshev polynomial series truncated to degree $K$:
    \begin{equation}
        f(x) \approx \sum_{k=0}^{K} c_k T_k(x)
    \end{equation}
    where $c_k$ represents the classically pre-computed Chebyshev spectral coefficients, and $T_k(x)$ is the $k$-th degree Chebyshev polynomial of the first kind satisfying the recurrence relation $T_{k+1}(x) = 2xT_k(x) - T_{k-1}(x)$.
    
    \item \textbf{Compute:} Each orthogonal polynomial degree $T_k(x)$ is isolated and mapped to an independent, ultra-shallow single-qubit functional rotation channel. The parameter $x$ is encoded directly into a parameterized state preparation angle $\theta = \arccos(x)$. Because the channels are isolated:
    \begin{equation}
        y_k = \langle \psi(\theta) | T_k | \psi(\theta) \rangle
    \end{equation}
    The individual quantum sub-components are bounded at a static, minimal circuit depth $D(T_k) = \mathcal{O}(1) \ll D_{\text{max}}$. This eliminates cumulative depth scaling entirely from the functional complexity.
    
    \item \textbf{Reduce:} Collect the scalar expectation outputs $y_k$ from the execution channels. The original non-linear evaluation is reconstructed by scaling the outputs with the extracted spectral coefficients over the classical addition monoid:
    \begin{equation}
        f(x) \approx \sum_{k=0}^{K} c_k y_k
    \end{equation}
    As a result, non-linear functional properties are evaluated with spatial isolation from hardware noise propagation.
\end{itemize}

\begin{figure}[htbp]
	\centering
	\resizebox{\columnwidth}{!}{%
		\begin{quantikz}
			\lstick{ancilla:} & \gate{H} \slice[style=dashed]{\text{\footnotesize\ensuremath{\lvert\Psi_1\rangle}}} & \ctrl{1} \slice[style=dashed]{\text{\footnotesize\ensuremath{\lvert\Psi_2\rangle}}} & \gate{H} \slice[style=dashed]{\text{\footnotesize\ensuremath{\lvert\Psi_{\text{final}}\rangle}}} & \meter{} \\
			\lstick{functional\_target:} & \qw & \gate{R_X(2k\theta)} & \qw & \qw
		\end{quantikz}%
	}
	\caption{Quantum circuit layout for the isolated $k$-th degree Chebyshev polynomial rotation primitive.}
	\label{fig:chebyshev_circuit_primitive}
\end{figure}
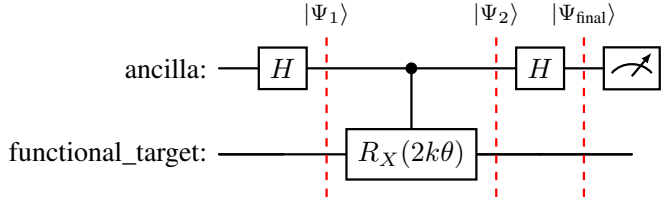

In Fig. \ref{fig:chebyshev_circuit_primitive}, we track the state vector evolution across each stage of the isolated $k$-th degree Chebyshev rotation channel:
\begin{align}
    \lvert \Psi_0 \rangle &= \lvert 0 \rangle_{\text{ancilla}} \otimes \lvert 0 \rangle_{\text{target}} \\
    \lvert \Psi_1 \rangle &= \frac{1}{\sqrt{2}}\bigl(\lvert 0 \rangle + \lvert 1 \rangle\bigr) \otimes \lvert 0 \rangle \\
    \lvert \Psi_2 \rangle &= \frac{1}{\sqrt{2}}\lvert 0 \rangle \otimes \lvert 0 \rangle + \frac{1}{\sqrt{2}}\lvert 1 \rangle \otimes R_X(2k\theta)\lvert 0 \rangle
\end{align}
Expanding the single-qubit functional rotation matrix action $R_X(2k\theta) = \cos(k\theta)I - i\sin(k\theta)X$ onto the target ground state yields:
\begin{equation}
    \lvert \Psi_2 \rangle = \frac{1}{\sqrt{2}}\lvert 00 \rangle + \frac{1}{\sqrt{2}}\lvert 1 \rangle \otimes \bigl(\cos(k\theta)\lvert 0 \rangle - i\sin(k\theta)\lvert 1 \rangle\bigr)
\end{equation}
Applying the final interference Hadamard gate to the auxiliary ancilla register collapses the superposition to:
\begin{equation}
\begin{split}
    \lvert \Psi_{\text{final}} \rangle &= \frac{1}{2}\lvert 0 \rangle \otimes \bigl[(1 + \cos(k\theta))\lvert 0 \rangle - i\sin(k\theta)\lvert 1 \rangle\bigr] \\
    &\quad + \frac{1}{2}\lvert 1 \rangle \otimes \bigl[(1 - \cos(k\theta))\lvert 0 \rangle + i\sin(k\theta)\lvert 1 \rangle\bigr]
\end{split}
\end{equation}

\begin{remark}{}
%The AOD framework projects a bounded non-linear functional mapping $f(x)$ onto a truncated Chebyshev polynomial series of degree $K$. 
The connection between the $k$-th degree Chebyshev polynomial of the first kind, $T_k(x)$, and the physical single-qubit rotation matrix $R_X(2k\theta)$ relies on the fundamental trigonometric identity:
\begin{equation}
    T_k(x) = \cos(k\theta), \quad \text{where } \theta = \arccos(x) \text{ for } x \in [-1, 1].
\end{equation}

\begin{enumerate}
\item The parameter $x$ is directly encoded into the quantum state preparation angle $\theta$. The single-qubit functional rotation matrix acting on the target register is parameterized by $2k\theta$ and expands via Euler's relation as:
\begin{equation}
    R_X(2k\theta) = \cos(k\theta)I - i\sin(k\theta)X.
\end{equation}

\item When applied to the target ground state $|0\rangle_{\text{target}}$, the operator yields:
\begin{equation}
    R_X(2k\theta)|0\rangle = \cos(k\theta)|0\rangle - i\sin(k\theta)|1\rangle.
\end{equation}

\item As detailed by the circuit primitive in Fig.~\ref{fig:chebyshev_circuit_primitive}, an auxiliary ancilla qubit is initialized in a uniform superposition via a Hadamard gate to form the state $|\Psi_1\rangle = \frac{1}{\sqrt{2}}(|0\rangle + |1\rangle) \otimes |0\rangle$. Applying the controlled-$R_X(2k\theta)$ gate generates the entangled state vector $|\Psi_2\rangle$.
%\begin{equation}
%    |\Psi_2\rangle = \frac{1}{\sqrt{2}}|00\rangle + \frac{1}{\sqrt{2}}|1\rangle \otimes \left( \cos(k\theta)|0\rangle - i\sin(k\theta)|1\rangle \right).
%\end{equation}
\item Applying the final Hadamard gate to the auxiliary ancilla register collapses the state vector to $|\Psi_{\text{final}}\rangle$.
%\begin{equation}
%    |\Psi_{\text{final}}\rangle = \frac{1}{2}|0\rangle \otimes \left[ (1 + \cos(k\theta))|0\rangle - i\sin(k\theta)|1\rangle \right] + \frac{1}{2}|1\rangle \otimes \left[ (1 - \cos(k\theta))|0\rangle + i\sin(k\theta)|1\rangle \right].
%\end{equation}
\item Measuring the auxiliary ancilla register in the computational Z-basis yields the targeted localized polynomial metric $y_k$ via the difference in outcome probabilities, isolating the exact analytical value of $T_k(x)$:
\begin{equation}
    y_k = P(0) - P(1) = \cos(k\theta) = T_k(x).
\end{equation}
\end{enumerate}
The total approximation error therefore contains both the Chebyshev truncation error and the quantum estimation error associated with noisy finite-shot evaluation of each $T_k(x)$.
\end{remark}

\subsection{Scenario 4: Convolutions}
%While standard 1-D vector space operators accommodate flat machine learning weights, 
High-dimensional computer vision operations such as spatial two-dimensional convolutions ($\text{Conv2D}$) traditionally present a severe depth bottleneck for NISQ systems due to massive sliding-window tensor contractions. 
%To extend the Algebraic Operator Decomposition (AOD) paradigm to spatial feature extraction, 
We use a parallelized quantum convolution layer mapped to an optimized classical reduction monoid using the \texttt{im2col} (Image-to-Column) transformation to expose the convolution as a collection of linear inner products to which AOD can be applied.

\begin{itemize}
    \item \textbf{Map:} Let the input tensor be an $N$-element spatial feature map $X \in \mathbb{R}^{C \times H \times W}$ and the targeted filter be a convolution kernel bank $K \in \mathbb{R}^{C_{\text{out}} \times C \times k_H \times k_W}$, where $C$ is the channel depth and $k_H \times k_W$ represents the kernel spatial footprint. The host uses a sliding window to extract every overlapping spatial coordinate neighborhood across the feature map. These sub-blocks are dynamically mapped to a vectorized patch matrix $A \in \mathbb{R}^{(H_{\text{out}} W_{\text{out}}) \times (C k_H k_W)}$, where the rows represent independent spatial translation coordinates and the columns aggregate flattened feature components. The convolution filters are flattened into a static kernel weights matrix $B \in \mathbb{R}^{(C k_H k_W) \times C_{\text{out}}}$. This maps the spatial stencil task into a single General Matrix Multiplication (GEMM) form:
    \begin{equation}
        Y_{\text{GEMM}} = A \times B
    \end{equation}
    
    \item \textbf{Compute:} Evaluate the exact inner products $\langle \psi^{(i)} \vert \tau_j \vert \psi^{(i)} \rangle$ for every row-column matrix coordinate intersection concurrently. Instead of compounding system noise by initializing a deep monolithic register, each row patch $A_{i,:}$ is L2-normalized and evaluated across independent, parallelized 2-qubit phase channels:
    \begin{equation}
        G_{i,j} = \text{Re}\langle \tilde{A}_{i,:} \vert \tilde{B}_{:,j} \vert \tilde{A}_{i,:} \rangle = P(0) - P(1)
    \end{equation}
    The physical circuit depth is strictly bounded at $D(\tau) = \mathcal{O}(1) \ll D_{\text{max}}$, removing spatial depth bottlenecks. By decoupling the joint noise map across discrete spatial coordinates ($\frac{\partial \Lambda_i(\rho_i)}{\partial \rho_j} = 0$), any active gate fault or depolarizing relaxation event inside spatial step $i$ remains trapped within that localized sub-register $\mathcal{H}_i$, satisfying the criteria for \textit{Spatial Error Isolation} (Theorem 1).
    
    \item \textbf{Reduce:} Scale the matrix elements directly against the pre-computed row-column L2 matrix norms matrix $S = \lVert A_{i,:} \rVert_2 \times \lVert B_{:,j} \rVert_2$ to restore the analytical magnitudes:
    \begin{equation}
        Y_{\text{scaled}} = G_{\text{quantum}} \odot S
    \end{equation}
    Finally, the reconstructed 2D array is reshaped into a structural 4D feature map layout:
    \begin{equation}
        Y \in \mathbb{R}^{\text{Batch} \times C_{\text{out}} \times H_{\text{out}} \times W_{\text{out}}}
    \end{equation}
    %This shifts the structural burden of spatial parameter tracking from fragile, deep quantum coherence memory over to a stable classical tensor addition reducer.
\end{itemize}
This experiment evaluates AOD as a noise/depth-management architecture rather than demonstrating end-to-end quantum computational advantage over classical convolution.

\begin{figure*}[t!]
    \centering
    \includegraphics[width=0.8\textwidth]{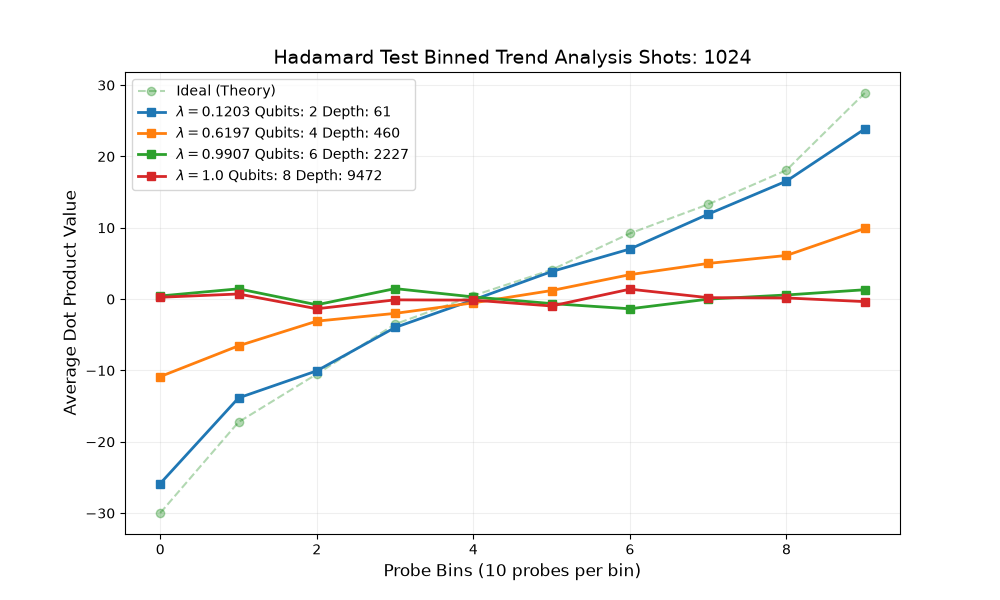}
    %\caption{Visual proof that dividing the feature space keeps the circuit safely below the noise threshold.}
	\caption{Visual Feature Space Division. For shallow partitions ($\lambda = 0.1203$, 2 Qubits, Depth: 61), the average dot product value tightly tracks the ideal baseline. As the circuit depth increases to $D = 460$ ($\lambda = 0.6197$, 4 Qubits), a noticeable flattening of the slope occurs. This highlights the early stages of information dissipation.}
    \label{fig:trend_analysis}
\end{figure*}
\begin{figure*}[t!]
    \centering
    \includegraphics[width=0.8\textwidth]{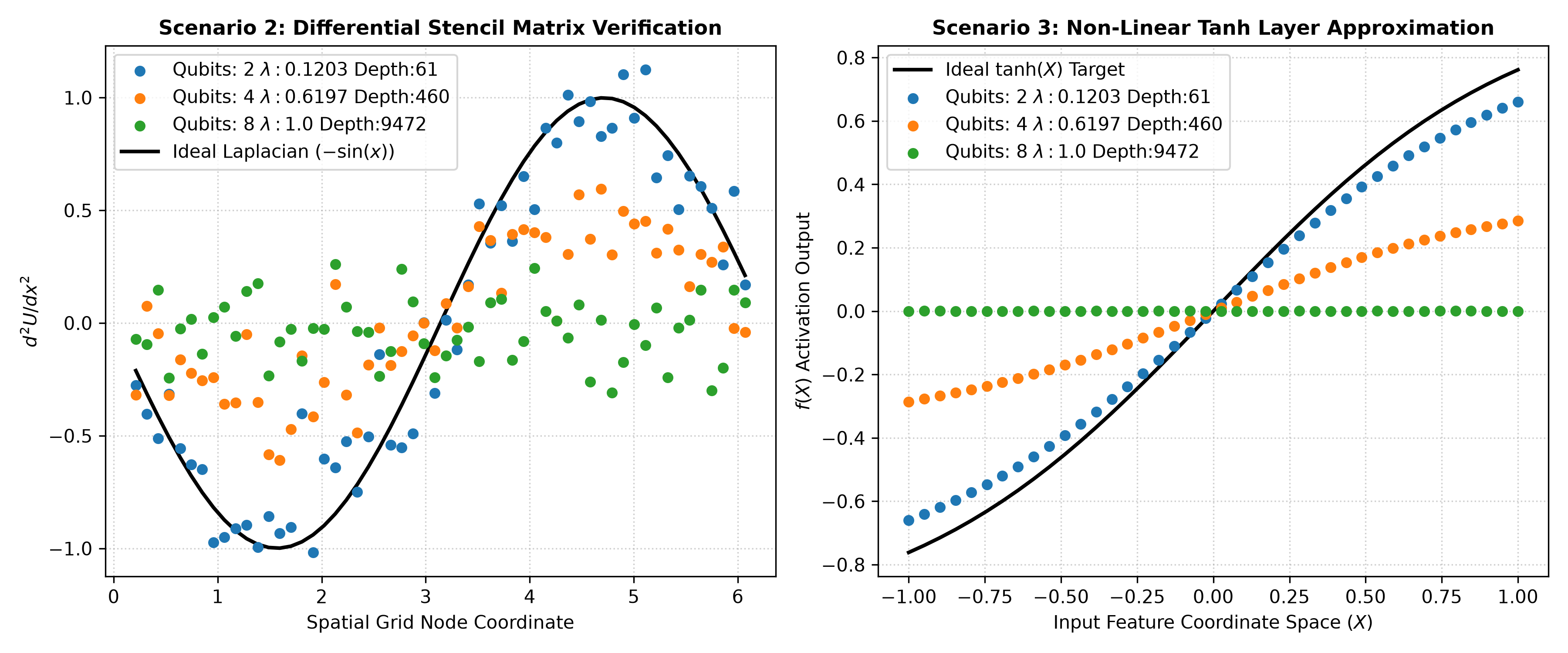}
	\caption{In the low-depth partition tier ($n=2$, $\lambda=0.1203$, $D=61$), both mathematical operators demonstrate excellent structural stability, tracking the precise contours of the exact analytical targets ($-\sin(x)$ and $\tanh(X)$) with minimal variance. At the deepest monolithic threshold ($n=8$, $\lambda=1.0$, $D=9472$), phase-decoherence completely dominates the system.}
    \label{fig:trend_analysis1}
\end{figure*}
\begin{figure*}[t!]
    \centering
    \includegraphics[width=0.8\textwidth]{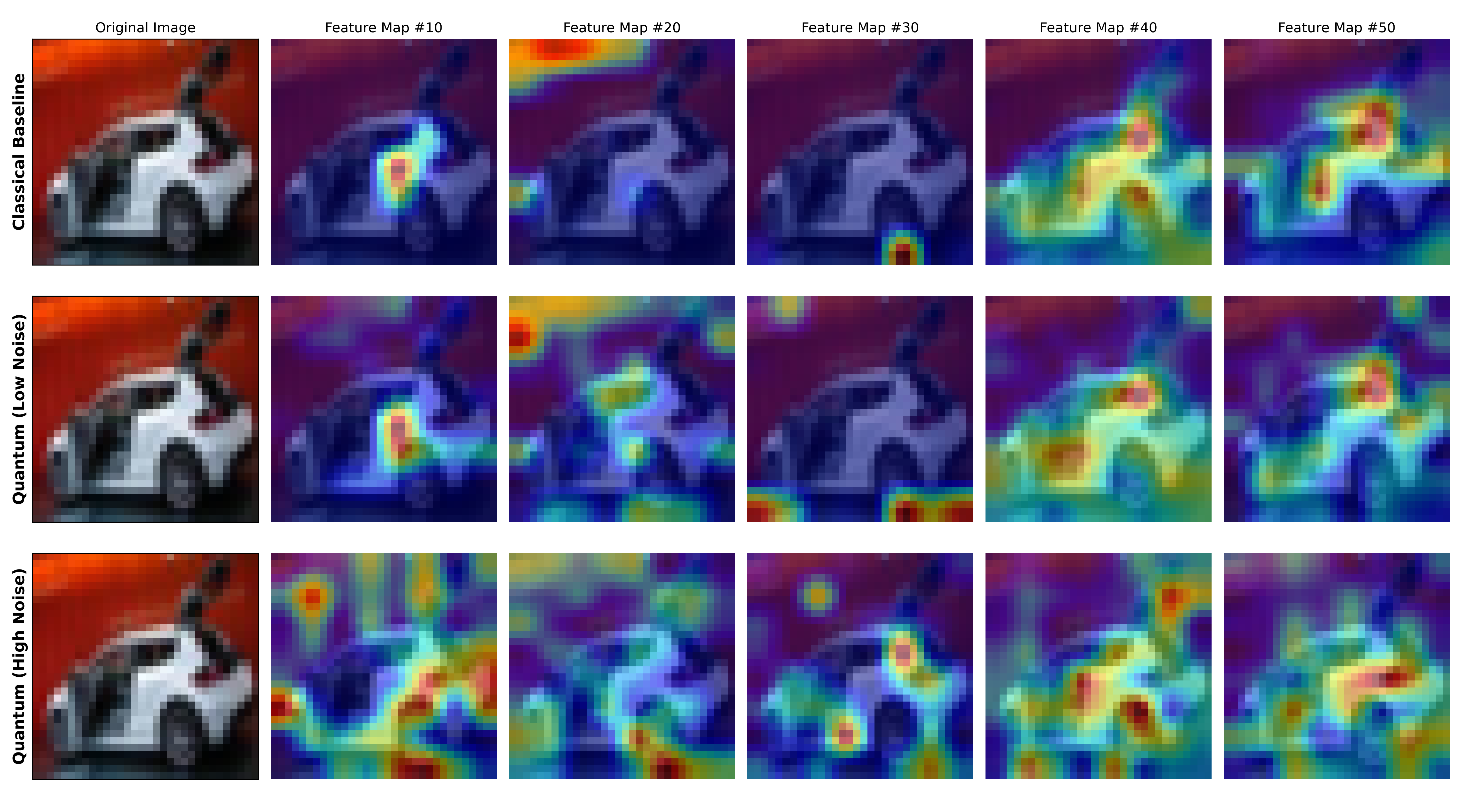}
	\caption{Quantum Convolution: forward inference pass on the first car image of CIFAR-10 under Ideal, Low Noise ($\lambda = 0.1203$, 2 Qubits, Depth: 61) and High Noise ($\lambda = 0.6197$, 4 Qubits, Depth: 460) settings using algorithms \ref{alg:hadamard_vectorized}, and \ref{alg:quantum_conv2d}. Note: Algorithm \ref{alg:hadamard_vectorized} emulates the expected measurement statistics of the Hadamard-test primitive under the adopted depolarizing/readout noise model.}
    \label{fig:conv2d_quantum}
\end{figure*}
\begin{figure}[htbp]
    \centering
    \includegraphics[width=\linewidth]{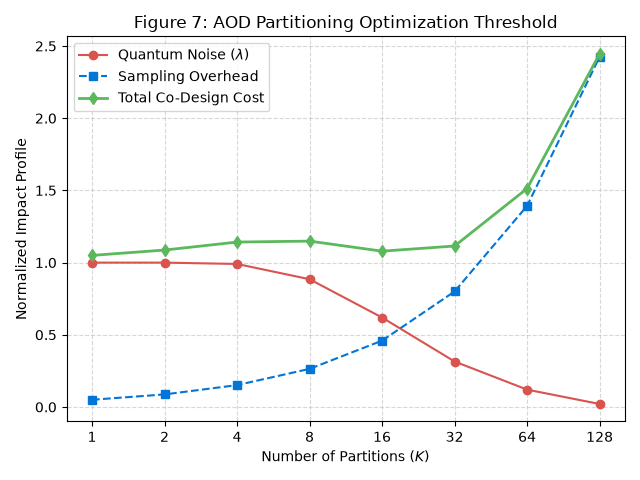}
    \caption{AOD partitioning tradeoff across $K$ channels. The optimization curve illustrates the competitive crossover front between quantum gate attenuation and classical sampling expansion factors.}
    \label{fig:tradeoff}
\end{figure}
\section{Simulation Results}
We begin with Scenario 1 for a linear transformation of an arbitrary batch of the MNIST-16 dataset and corresponding random weights of shapes $X(1,256), W(256, 128)$ (Fig. \ref{fig:trend_analysis}).
The resulting vectors are sorted by the ideal indices and binned in groups of 10 to cleanup the curves.
Our results show that, under the selected hardware error parameters (Table \ref{tab:circuit_complexity_mapreduce}), increasing circuit depth from 61 to 460 produces substantial signal flattening, while the depth-9472 case approaches complete depolarization in the adopted noise model.
%Our results show that, at depths exceeding the {\decot} $D \approx 500$, the quantum state undergoes complete phase-decoherence.
The contrast between the 2-qubit and 8-qubit fidelity response empirically validates our AOD paradigm: fracturing a deep monolithic channel into small, shallow sub-channels provides a practical mechanism for maintaining individual circuit executions within the lower-noise regime of the adopted model.

\subsection{Cross-Workload Behavior Under the Noise Model}
The paradigm is further validated in Figure \ref{fig:trend_analysis1}, across both spatial differential stencil matrix operations (Scenario 2) and non-linear polynomial activation loops (Scenario 3): signal dissipation shows the same decay profile seen in monolithic inner products, establishing a noise baseline.
As the partition registers expand to a moderate configuration ($n=4$, $\lambda=0.6197$, $D=460$), a uniform flattening of the amplitude response emerges, signifying initial information loss where depolarizing noise contracts state distributions toward a coin-flip probability threshold ($p_0 \to 0.5$).
In summary, similar depth-dependent attenuation is observed across the differential and nonlinear examples.

\subsection{Scenario 4: Convolutions}
Fig. \ref{fig:conv2d_quantum} shows the impact of the {\decot} on spatial feature extractions:
\begin{itemize}
	\item \textbf{Row 1 (Classical Baseline):} Matches current localized hot-spots tracking the lines of the car door and roofline precisely.
	\item \textbf{Row 2 (Quantum Low Noise):} The spatial contours look slightly fuzzy due to shot sampling noise, but clearly preserve the overall structural geometry of the car.
	%\item \textbf{Row 3 (Quantum High Noise):} The heatmap highlights bleed randomly across the image frame or entirely flatten into a uniform dull hue. This visualizes Spatial Isolation (Theorem 1) empirically: showing that once a register breaks past the {\decot}, the state decays toward a coin-flip probability, wiping out edge-detection features.
	\item \textbf{Row 3 (Quantum High Noise):} The heatmap highlights bleed randomly across the image frame or entirely flatten into a uniform dull hue. This illustrates the practical consequence of the depth-dependent noise model for the reconstructed spatial output.
\end{itemize}

\begin{algorithm}
\scriptsize
\caption{Vectorized Hadamard Test Emulator with Gate Noise and Readout Error}
\label{alg:hadamard_vectorized}
\begin{algorithmic}[1] % [1] enables line numbering

\Function{dot\_2D\_vectorized}{$A, B$}
    \State \textbf{Input:} Matrix $A \in \mathbb{C}^{R \times K}$, Matrix $B \in \mathbb{C}^{K \times C}$
    \State \textbf{Output:} Simulated noisy scaled quantum expectation matrix
    
    \State $N_{\text{shots}} \gets \text{\_\_dot\_sim\_shots}$ \Comment{Global configuration parameters}
    \State $\lambda \gets \text{\_\_noise\_level}$
    \State $\epsilon_{\text{read}} \gets \text{\_\_readout\_error}$
    
    \Statex
    \State // 1. Compute exact dot products via matrix multiplication
    \State $D \gets A \times B$
    
    \Statex
    \State // 2. Compute matrix norms for quantum state scaling
    \State $V_A \gets \begin{bmatrix} \|A_{1,:}\|_2 & \|A_{2,:}\|_2 & \dots & \|A_{R,:}\|_2 \end{bmatrix}^T$
    \State $V_B \gets \begin{bmatrix} \|B_{:,1}\|_2 & \|B_{:,2}\|_2 & \dots & \|B_{:,C}\|_2 \end{bmatrix}$
    \State $S \gets V_A \times V_B$ \Comment{Scale matrix of shape $R \times C$}
    
    \Statex
    \State // 3. Calculate ideal probability $P_0$ for Hadamard test
    \If{$S_{i,j} \neq 0$}
        \State $\bar{D}_{i,j} \gets \frac{\text{Re}(D_{i,j})}{S_{i,j}}$
    \Else
        \State $\bar{D}_{i,j} \gets 0$
    \EndIf
    \State $P_0 \gets \frac{1}{2} \left( \mathbf{1} + \bar{D} \right)$
    
    \Statex
    \State // 4. Inject gate depolarizing noise and readout error
    \State $P_0^{\text{gate}} \gets (1 - \lambda) P_0 + \frac{\lambda}{2}$
    \State $p(0|0) \gets 1 - \epsilon_{\text{read}}$
    \State $p(0|1) \gets \epsilon_{\text{read}}$
    \State $P_0^{\text{hw}} \gets \left( P_0^{\text{gate}} \odot p(0|0) \right) + \left( (\mathbf{1} - P_0^{\text{gate}}) \odot p(0|1) \right)$
    
    \Statex
    \State // 5. Batch stochastic sampling over hardware probability matrix
    \State $P_0^{\text{clipped}} \gets \min(\max(P_0^{\text{hw}}, \mathbf{0}), \mathbf{1})$
    \State $N_0 \sim \text{Binomial}(n=N_{\text{shots}}, \, p=P_0^{\text{clipped}})$
    
    \Statex
    \State // 6. Reconstruct expectation value and re-scale
    \State $E \gets 2 \left( \frac{N_0}{N_{\text{shots}}} \right) - \mathbf{1}$
    \State \textbf{return} $E \odot S$
\EndFunction

\end{algorithmic}
\end{algorithm}

\begin{algorithm}[htbp]
\scriptsize
\caption{Noisy Quantum 2D Convolution via Vectorized \texttt{im2col} Mapping}
\label{alg:quantum_conv2d}
\begin{algorithmic}[1]
\Require Input feature map $X \in \mathbb{R}^{C \times H \times W}$, Convolution kernel $K \in \mathbb{R}^{C_{\text{out}} \times C \times k_H \times k_W}$, stride $s$, padding $p$
\Ensure Output tensor $Y \in \mathbb{R}^{C_{\text{out}} \times H_{\text{out}} \times W_{\text{out}}}$
\Statex

\Function{Conv2D\_Quantum}{$X, K, s, p$}
    \State // 1. Enforce uniform shape profiles for missing channel dimensions
    \If{$\text{ndim}(X) = 2$} \Comment{Shape to $(1, H, W)$}
        \State $X \gets X[\text{NewAxis}, :, :]$ 
    \EndIf
    \If{$\text{ndim}(K) = 2$} \Comment{Shape to $(1, 1, k_H, k_W)$}
        \State $K \gets K[\text{NewAxis}, \text{NewAxis}, :, :]$ 
    \EndIf
    
    \State Extract configurations: $C, H_{\text{in}}, W_{\text{in}} \gets \text{shape}(X)$ and $C_{\text{out}}, \_, k_H, k_W \gets \text{shape}(K)$
    \Statex
    \State // 2. Apply boundary zero-padding configurations
    \If{$p > 0$}
        \State $X \gets \text{Pad}(X, \text{pad}=p, \text{mode}=\text{'constant'})$ 
        \State $H_{\text{in}}, W_{\text{in}} \gets \text{shape}(X)$
    \EndIf
    \Statex
    \State // 3. Compute resolution bounds for target feature map
    \State $H_{\text{out}} \gets \lfloor (H_{\text{in}} - k_H)/s \rfloor + 1$
    \State $W_{\text{out}} \gets \lfloor (W_{\text{in}} - k_W)/s \rfloor + 1$
    \Statex
    \State // 4. Map Stage: Vectorized patch isolation over sliding windows
    \State $X_{\text{windows}} \gets \text{SlidingWindowView}(X, \text{window}=(k_H, k_W))[:, \, ::s, \, ::s, \, :, \, :]$
    \Statex
    \State // 5. Flatten sub-tensors into linear 2D operator matrices
    \State $A \gets \text{Reshape}(X_{\text{windows}}, \text{shape}=(H_{\text{out}} \cdot W_{\text{out}}, \, C \cdot k_H \cdot k_W))$
    \State $B \gets \text{Reshape}(K, \text{shape}=(C_{\text{out}}, \, C \cdot k_H \cdot k_W))^T$
    \Statex
    \State // 6. Execute the Noisy Hadamard Matrix Multiplication
    \State $G_{\text{quantum}} \gets \text{\textsc{Dot\_2D\_Vectorized}}(A, B)$ 
    \Statex
    \State // 7. Reshape the 2D output to 3D Out Features (Out\_C, h\_out, w\_out)
    \State $Y \gets \text{Reshape}(G_{\text{quantum}}^T, \text{shape}=(C_{\text{out}}, \, H_{\text{out}}, \, W_{\text{out}}))$
    
    \State \Return $\text{Squeeze}(Y)$
\EndFunction
\end{algorithmic}
\end{algorithm}

\section{Discussion}
\subsection{Co-Design Optimization: The AOD Partitioning Tradeoff}

The practical viability of the AOD paradigm hinges on balancing quantum hardware resilience against classical reconstruction overhead. As evaluated in Fig.~\ref{fig:tradeoff}, configuring the partition slicing layout yields two distinct operational regimes:
\begin{itemize}
    \item \textbf{The Monolithic/Low-K Regime ($K \le 4$):} Circuits feature wide register allocations and deep instruction layers. Here, spatial gate errors propagate unchecked, driving the cumulative depolarizing noise parameter $\lambda \to 1.0$ and completely wiping out the signal.
    \item \textbf{The Highly-Partitioned Regime ($K \ge 64$):} While sub-circuits are ultra-shallow ($D \le 61$) and near-ideal, the classical sampling budget expands quadratically according to the 1-norm coefficient tax ($\|c\|_1^2$) derived in Theorem 2.
\end{itemize}
An intermediate configuration ($K=16$ to $K=32$) shows a clear co-design sweet spot, minimizing total computational friction and giving a structured operational framework for NISQ execution.

\subsection{Formalization of the Co-Design Partition Selection Problem}
To move beyond a qualitative heuristic, we formalize the partition topology selection as a constrained multi-objective co-design optimization problem. Let $\mathcal{K}$ be the set of valid partition layouts. A system architect seeks to find the optimal partitioning granularity $K^*$ that minimizes total operational cost while strictly respecting a hardware-imposed or application-defined error bound $\epsilon_{\max}$:

\begin{equation}
    K^* = \arg\min_{K \in \mathcal{K}} \mathcal{C}_{\text{total}}(K) \quad \text{subject to} \quad \mathcal{E}_{\text{total}}(K) \le \epsilon_{\max}
\end{equation}

We define the total operational cost model $\mathcal{C}_{\text{total}}(K)$ as a weighted combination of quantum resource execution metrics and classical host overhead processing burdens:
\begin{equation}
    \mathcal{C}_{\text{total}}(K) = w_q \cdot S_{\text{total}}(K) \cdot \bar{D}(K) + w_c \cdot \mathcal{T}_{\text{classical}}(K)
\end{equation}
where $S_{\text{total}}(K)$ is the optimal shot allocation derived from Theorem 2, $\bar{D}(K)$ is the mean gate depth of the resulting sub-circuits, $\mathcal{T}_{\text{classical}}(K)$ is the classical wall-clock time required to execute the commutative linear recombination monoid, and $w_q, w_c$ are user-defined optimization weight parameters reflecting local infrastructure priorities (e.g., QPU queue times vs. CPU cluster availability).

Concurrently, the global reconstruction error $\mathcal{E}_{\text{total}}(K)$ aggregates both physical gate decoherence profiles and classical sampling variance components:
\begin{equation}
    \mathcal{E}_{\text{total}}(K) \le \Delta_T(K) + \frac{\|c(K)\|_1}{\sqrt{S_{\text{total}}(K)}}
\end{equation}
where $\Delta_T(K)$ represents the deterministic physical error bound defined by Theorem 1, and the second term accounts for the stochastic statistical sampling uncertainty dictated by the central limit theorem. 

By bounding $\mathcal{E}_{\text{total}}(K) \le \epsilon_{\max}$, the constraint dynamically disqualifies low-$K$ layouts whose physical depth violates the hardware coherence envelope ($D > D_{\max}$), while the objective function naturally penalizes high-$K$ configurations due to their severe 1-norm sampling expansion factors. 
This formalization positions the choice of partition layout as a standard, co-design trade-off landscape for algorithmic execution.

\subsection{Architectural Assumptions}
Although, our architecture has a proactive philosophy with solid mathematical grounding in Cauchy’s Functional Equation and the Riesz Representation Theorem, it makes several assumptions:
\begin{itemize}
	\item \textbf{QPU Hardware:} It heavily assumes access to an asynchronous, parallel multi-QPU network to circumvent time constraints.
	Even though today's classical latency overhead of distributing and collecting these jobs (e.g., REST API queuing times or cloud job scheduling bottlenecks), dwarf raw quantum execution times,
	these latencies can be mitigated with platforms such as IBM Qiskit Serverless which unify classical and quantum resources hybrid workflows without managing underlying infrastructure \cite{qiskit_serverless}.
	\item \textbf{Sampling Overhead \& Practical Bottlenecks:} While AOD evades exponential phase-decoherence by fracturing deep circuits, it introduces a pronounced sampling overhead governed by the profile of the decomposition. As formalized in Theorem 2, the total hardware shot budget scales quadratically with the 1-norm of the classical coefficient vector ($\|c\|_1^2$). 
	%For high-dimensional dense matrices, complex differential stencils with fine grid resolutions, or non-linear functions requiring high-degree Chebyshev expansions, this ``1-norm tax'' can grow rapidly, leading to a significant measurement bottleneck. 
	Although this ``1-norm tax'' can grow rapidly, it can be controlled by adjusting the partition size: Larger partitions reduce the number of circuits and the total shot budget. However, they also create deeper circuits that suffer from higher noise and decoherence.
	\item \textbf{Spatial Error Containment:} Although errors are confined to sub-chanel $(\tau_i)$, this does not mean the global computation is protected from that error.
	If $T=\sum_k c_k\tau_k$ and one measurement $(y_i)$ is biased by hardware noise, then the reconstructed result
	$\hat T=\sum_k c_k y_k$ is still biased by approximately $(c_i\Delta_i)$.
	AOD contains the error \textbf{between channels}, it does not eliminate its contribution to the final observable. This is formalized in Theorem 1.
\end{itemize}

%\textit{In summary, the practical viability of AOD hinges on the partition size that best minimizes noise and decoherence while wrangling the total sampling overhead in a multi-QPU setup.}
\section{Conclusion and Future Work}
%In this work, we presented a unified anti-decoherence and spatial error-mitigation framework using Algebraic Operator Decomposition (AOD) and resource-adaptive parallel Hadamard Test stacking. 
%Furthermore, we showed that the AOD paradigm seamlessly extends beyond linear matrix transformations to non-linear layers and spatial differential transformations. By mapping complex operators onto independent, ultra-shallow sub-channels, our architecture systematically subverts the physical constraints of the NISQ era, bounding quantum execution time entirely below the physical decoherence threshold. Our multi-lambda benchmarking empirically validated the existence of a universal "{\decot}" near a gate depth of $D \approx 500$, illustrating that fine-grained space partitioning is not merely an optimization, but a physical necessity to prevent quantum expectations from dissipating into noise.

AOD provides an operator-level framework for transforming a deep quantum computation into a collection of shallower, independently executable subchannels followed by classical reconstruction. Under the depolarizing model studied here, the resulting reduction in maximum circuit depth can substantially reduce noise accumulation, at the cost of additional classical computation and sampling. The results across inner products, differential operators, nonlinear approximations, and convolution demonstrate the generality of the decomposition principle, while the optimal partition size remains hardware- and workload-dependent.

\subsection{Future Research Directions}
While this study establishes the foundations of noise mitigation via operator fracturing, several avenues of exploration remain:

\begin{itemize}
    \item \textbf{Integration with Transformer Architectures:} Given that the self-attention mechanism in modern Large Language Models (LLMs) is fundamentally dominated by $\mathcal{O}(N^2)$ matrix operations, our primary immediate frontier focuses on accelerating the $QK^T$ attention-score calculation. By mapping attention matrices to large vertical Hadamard Test stacks, the time complexity scales relative to the stacking factor $K$, theoretically enabling context-window expansion far beyond classical memory-bandwidth limitations.
    %\item \textbf{Crosstalk Mitigation and Real-Hardware Schedulers:} Massive vertical stacking inherently increases the risk of physical inter-qubit crosstalk and multi-register thermal relaxation. Future iterations of our compiler pipeline will focus on developing hardware-aware spatial schedulers that dynamically optimize the balancing pattern, trading off parallel register density against the cumulative global system error rate.
    \item \textbf{High-Dimensional Tensor Scaling:} We aim to scale our vectorized benchmarks to complex computer vision spaces, such as CIFAR-10 and ImageNet, systematically tracking the resource-adaptive limits of the feature registers when processing dense tensors.
\end{itemize}

\appendix

\newtheorem{theorem}{Theorem}
\newtheorem{proofsketch}{Proof Sketch}

\subsection{Theoretical Foundations of Spatial Error Containment}
To rigorously justify the claim of error containment under Algebraic Operator Decomposition (AOD), we formalize the mapping of error channels across partitioned registers.

\begin{theorem}[AOD Reconstruction Error Bound Under Independent Local Noise]
Let $T \in \mathcal{L}(V, W)$ be a global linear operator decomposed into $K$ independent sub-channels such that $T = \sum_{k=1}^K c_k \tau_k$. Let $\Lambda_k$ represent the independent local noise channels operating on sub-registers of qubit width $n_k$ and depth $D_k \ll D_{\max}$. Assuming the global noisy channel factorizes cleanly across the independent hardware executions as $\Lambda = \bigotimes_k \Lambda_k$, the absolute deviation $\Delta_T = |\langle T \rangle_\Lambda - \langle T \rangle_{\text{ideal}}|$ between the noisy reconstructed expectation value and the ideal expectation value satisfies:
\begin{equation}
    \Delta_T \le 2\sum_{k=1}^K |c_k| \cdot \|\tau_k\|_\infty \cdot \mathcal{D}\left(\Lambda_k(\rho_k), \rho_{k,\text{ideal}}\right)
\end{equation}
where $\mathcal{D}(\rho, \sigma) = \frac{1}{2}\text{Tr}\sqrt{(\rho-\sigma)^\dagger(\rho-\sigma)}$ denotes the conventional quantum state trace distance.
\end{theorem}

\begin{proof}
Let the ideal state preparation and execution channel for an arbitrary sub-component $\tau_k$ be represented by the dense pure state density matrix $\rho_{k,\text{ideal}} = |\psi_k\rangle\langle\psi_k|$. In the presence of local hardware noise, the execution on an isolated $n_k$-qubit QPU register transforms the target state via a completely positive trace-preserving (CPTP) map $\Lambda_k(\rho_k)$. We express this local noise profile under a generalized depolarizing Kraus representation:
\begin{equation}
    \Lambda_k(\rho_k) = (1 - \lambda_k)\rho_{k,\text{ideal}} + \lambda_k \frac{I}{2^{n_k}}
\end{equation}
where $\lambda_k = 1 - (1 - \epsilon_g)^{n_k \cdot D_k}$ defines the localized volumetric noise parameter, and $I/2^{n_k}$ represents the maximally mixed state on the $n_k$-qubit local Hilbert space.

By invoking the linearity property of the trace operator, the macro-reconstructed expectation value gathered by the classical CPU reducer is $\langle T \rangle_\Lambda = \sum_{k=1}^K c_k \text{Tr}\left(\tau_k \Lambda_k(\rho_k)\right)$. The absolute deviation from the true analytical expectation value is bounded using the triangle inequality over the classical scalar fields:
\begin{align}
    \Delta_T &= \left| \sum_{k=1}^K c_k \text{Tr}\left(\tau_k \Lambda_k(\rho_k)\right) - \sum_{k=1}^K c_k \text{Tr}\left(\tau_k \rho_{k,\text{ideal}}\right) \right| \\
    &\le \sum_{k=1}^K |c_k| \cdot \left| \text{Tr}\left(\tau_k \left[ \Lambda_k(\rho_k) - \rho_{k,\text{ideal}} \right]\right) \right|
\end{align}

By applying the Hölder inequality for matrix operators, the trace of the product is bounded by the product of the operator norm (spectral radius) and the trace class norm:
\begin{equation}
    \left| \text{Tr}\left(\tau_k \left[ \Lambda_k(\rho_k) - \rho_{k,\text{ideal}} \right]\right) \right| \le \|\tau_k\|_\infty \cdot \text{Tr}\left(\left| \Lambda_k(\rho_k) - \rho_{k,\text{ideal}} \right|\right)
\end{equation}

Recalling that the trace distance convention is defined as $\mathcal{D}(\rho, \sigma) = \frac{1}{2}\|\rho - \sigma\|_{\text{tr}} = \frac{1}{2}\text{Tr}(|\rho - \sigma|)$, we substitute this definition into the expression:
\begin{equation}
    \left| \text{Tr}\left(\tau_k \left[ \Lambda_k(\rho_k) - \rho_{k,\text{ideal}} \right]\right) \right| \le 2 \|\tau_k\|_\infty \cdot \mathcal{D}\left(\Lambda_k(\rho_k), \rho_{k,\text{ideal}}\right)
\end{equation}

Since all partitioned primitives $\tau_k$ correspond to bounded unitary operations or normalized projections (such as the standard 1-qubit and 2-qubit Hadamard Test layouts), their operator norm is strictly bounded by unity, $\|\tau_k\|_\infty \le 1$. Substituting this structural normalization back into the linear summation yields the desired upper bound.

Under the architectural factorization assumption ($\Lambda = \bigotimes_k \Lambda_k$), the noisy execution of any given sub-channel $\tau_i$ takes place within an isolated tensor block. It possesses no physical mechanism to alter or cascade into the localized quantum state transitions of a separate sub-channel $\tau_j$. The global mathematical error is thus cleanly trapped as a stable, un-entangled linear sum of isolated local variances, completing the proof.
\end{proof}

% -- Theorem 2
\subsection{The Complexity Crossover and Sampling Cost of AOD Reconstruction}
While the AOD framework contains exponential gate error propagation by fracturing monolithic operators into shallow sub-channels, it introduces a classical sampling overhead governed by the scale of the decomposition coefficients. 

\begin{theorem}[Sampling Cost of AOD Reconstruction]
Let $c = [c1, c2, \dots, c_K]^T$ be the vector of classical coefficients extracted during the operator decomposition phase. To resolve the global expectation value $\langle T \rangle$ to an absolute statistical precision error of $\epsilon$ with a confidence interval of $1-\delta$ using independent hardware estimators bounded in $[-1, 1]$, the total hardware shot budget $S_{\text{total}}$ required by the AOD commutative monoid scales as:
\begin{equation}
    S_{\text{total}} = \frac{2\|c\|_1^2}{\epsilon^2} \ln \left(\frac{2}{\delta}\right)
\end{equation}
\end{theorem}

\begin{proof}
Let $\hat{Y}_k$ be the empirical estimator for the local expectation value $\langle \tau_k \rangle$ obtained by averaging $S_k$ independent hardware measurement shots on an ultra-shallow template register. Because each measurement outcome yields an eigenvalue or expectation difference bounded by the interval $[-1, 1]$, the range of each independent random variable $c_k \hat{Y}_k$ is $2|c_k|$. The global AOD estimator is reconstructed via the classical linear combination $\hat{T}_{\text{AOD}} = \sum_{k=1}^K c_k \hat{Y}_k$.

By applying the Hoeffding inequality for independent bounded random variables, the probability that the macro-reconstructed estimator deviates from its true phase-averaged mean by more than $\epsilon$ is strictly bounded by:
%\begin{equation}
%    P\left( \left| \hat{T}_{\text{AOD}} - \mathbb{E}[\hat{T}_{\text{AOD}}] \right| \ge \epsilon \right) \le 2 \exp \left( - \frac{2\epsilon^2}{\sum_{k=1}^K \frac{(2|c_k|)^2}{S_k}} \right) = 2 \exp \left( - \frac{\epsilon^2}{2 \sum_{k=1}^K \frac{c_k^2}{S_k}} \right)
%\end{equation}
\begin{equation}
\begin{aligned}
    P\left( \left| \hat{T}_{\text{AOD}} - \mathbb{E}[\hat{T}_{\text{AOD}}] \right| \ge \epsilon \right) &\le 2 \exp \left( - \frac{2\epsilon^2}{\sum_{k=1}^K \frac{(2|c_k|)^2}{S_k}} \right) \\
    &= 2 \exp \left( - \frac{\epsilon^2}{2 \sum_{k=1}^K \frac{c_k^2}{S_k}} \right)
\end{aligned}
\end{equation}

To minimize the global failure probability given a fixed total shot budget $S_{\text{total}} = \sum S_k$, we apply the method of Lagrange multipliers to minimize the variance expression in the denominator, yielding an optimal shot allocation strategy where $S_k = S_{\text{total}} \frac{|c_k|}{\|c\|_1}$. Substituting this optimal allocation back into the sum simplifies the denominator constraint:
\begin{equation}
    \sum_{k=1}^K \frac{c_k^2}{S_k} = \sum_{k=1}^K \frac{c_k^2}{S_{\text{total}} \frac{|c_k|}{\|c\|_1}} = \frac{\|c\|_1}{S_{\text{total}}} \sum_{k=1}^K |c_k| = \frac{\|c\|_1^2}{S_{\text{total}}}
\end{equation}

Substituting this reduction back into the exponential bound simplifies the expression to:
\begin{equation}
    2 \exp \left( - \frac{\epsilon^2 S_{\text{total}}}{2\|c\|_1^2} \right)
\end{equation}
Setting this upper bound equal to our maximum allowable significance threshold $\delta$ and isolating $S_{\text{total}}$ directly yields the precise sampling cost equation, resolving the structural coefficient.
\end{proof}

\begin{corollary}[Asymptotic Efficiency Crossover Boundaries]
Let $S_{\text{mono}}(D)$ be the hardware shot budget required to resolve an equivalent monolithic quantum circuit of depth $D$ to identical precision $\epsilon$ under a uniform depolarizing rate $\epsilon_g$. In the asymptotic limit where circuit depth dramatically exceeds the physical hardware coherence threshold ($D \gg D_{\max} \approx T_2/t_{\text{gate}}$), complete entropic depolarization drives the state vector toward a maximally mixed state, causing the monolithic shot budget to diverge:
\begin{equation}
    \lim_{D \to \infty} S_{\text{mono}}(D) = \lim_{D \to \infty} \frac{1}{\epsilon^2 (1 - \epsilon_g)^{2nD}} = \infty
\end{equation}
Consequently, for any finite operator footprint characterized by a bounded norm $\|c\|_1 < \infty$, there exists a critical circuit depth threshold $D_{\text{crit}}$ past which the classical sampling overhead of AOD can become more favorable in the modeled resource regime.
\end{corollary}

\bibliographystyle{plain} % Or 'unsrt', 'alpha', 'ieeetr'
\bibliography{paper_references}

\end{document}